\documentclass[a4paper]{article}

\usepackage[utf8]{inputenc} 
\usepackage{textcomp} 
\usepackage{graphicx}  
\usepackage{flafter}  

\usepackage{amsmath,amssymb,amsfonts,amsthm}  
\usepackage{bm}  
\usepackage{bbm}
\usepackage{mathtools}
\usepackage{enumitem}
\usepackage{authblk}
\mathtoolsset{showonlyrefs}
\usepackage[numbers]{natbib}

\usepackage{verbatim}
\usepackage{xcolor}

\usepackage[bookmarks,colorlinks,breaklinks]{hyperref}  
\definecolor{dullmagenta}{rgb}{0.4,0,0.4}   
\definecolor{darkblue}{rgb}{0,0,0.4}
\hypersetup{linkcolor=red,citecolor=blue,filecolor=dullmagenta,urlcolor=darkblue} 
\usepackage{memhfixc}  
\usepackage{pdfsync}  

\newtheorem{theorem}{Theorem}
\newtheorem{lemma}{Lemma}
\newtheorem{remark}{Remark}

\newtheorem{definition}{Definition}
\newtheorem{proposition}{Proposition}

\newtheorem{innercustomgeneric}{\customgenericname}
\providecommand{\customgenericname}{}
\newcommand{\newcustomtheorem}[2]{%
\newenvironment{#1}[1]
{%
 \renewcommand\customgenericname{#2}%
 \renewcommand\theinnercustomgeneric{##1}%
 \innercustomgeneric
}
{\endinnercustomgeneric}
}

\newcustomtheorem{customthm}{Theorem}
\newcustomtheorem{customlemma}{Lemma}

\usepackage{braket}

\def\tr{{ \rm tr}}

\def\cE{\mathcal E}

\def\cN{\mathcal N}

\def\id{\mathbbm{1}}

\newcommand{\proj}[1]{|#1\rangle\langle#1|}

\newcommand*{\cL}{\mathcal{L}}

\newcommand*{\cO}{\mathcal{O}}
\newcommand*{\cP}{\mathcal{P}}

\usepackage{epstopdf}

\usepackage{xspace}

\begin{document}

\title{The conditional entropy power inequality\\ for quantum additive noise channels}
\author{Giacomo De Palma and Stefan Huber}
\date{\today}
\maketitle
\begin{abstract}
We prove the quantum conditional entropy power inequality for quantum additive noise channels. This inequality lower bounds the quantum conditional entropy of the output of an additive noise channel in terms of the quantum conditional entropies of the input state and the noise when they are conditionally independent given the memory. We also show that this conditional entropy power inequality is optimal in the sense that we can achieve equality asymptotically by choosing a suitable sequence of Gaussian input states. We apply the conditional entropy power inequality to find an array of information-theoretic inequalities for conditional entropies which are the analogs of inequalities which have already been established in the unconditioned setting. Furthermore, we give a simple proof of the convergence rate of the quantum Ornstein-Uhlenbeck semigroup based on entropy power inequalities.
\end{abstract}

\section{Introduction}
Additive noise channels are central objects of interest in information theory.
A general class of such channels can be modeled by the well-known convolution operation:
If $X$ and $Y$ are two independent random variables with values in $\mathbb{R}^k$, the convolution operation $(X,Y) \mapsto X+Y$
combines $X$ and $Y$ into a new random variable $X+Y$, the probability density function of which is given by
\begin{equation}
  f_{X+Y}(z) := \int_{\mathbb{R}^{k}} f_X(z-x) f_Y(x)\; \text{d}^k x\ .
  \label{eq:conv}
\end{equation}
The convolution is a well-studied operation and it plays a role in many inequalities from functional analysis, such as Young's Inequality and its sharp version~\cite{Beckner75,BrascampLieb76} as well as the entropy power inequality~\cite{Shannon48-2,Stam59,Blachman65,Demboetal91}. These inequalities have important applications in classical information theory, as they can be used to bound communication capacities, which was originally carried out by Shannon~\cite{Shannon48-2}. An extensive overview of the many related inequalities in this area is given in~\cite{Demboetal91}.

Central to the work presented here is the entropy power inequality. It deals with the entropy of a linear combination of two independent random variables $X$ and $Y$ with values in $\mathbb{R}^k$,
\begin{equation}
 Z := \sqrt{\lambda} X + \sqrt{|1-\lambda|} Y, \qquad \lambda \geq 0\ .
 \label{eq:clcombination}
\end{equation}
The statement of the entropy power inequality~\cite{Shannon48-2,Stam59,Blachman65,Demboetal91} is
\begin{equation}
  \exp{\frac{2 S(Z)}{k}} \geq \lambda \exp{\frac{2S(X)}{k}} + | 1-\lambda | \exp{\frac{2S(Y)}{k}}\ ,
  \label{eq:clepi}
\end{equation}
where $S(X)
$ is the Shannon differential entropy of the random variable $X$.
A conditional version of \eqref{eq:clepi} can easily be derived: If $X$ and $Y$ are conditionally independent given the random variable $M$ (sometimes interpreted as a memory), then
\begin{equation}
  \exp{\frac{2 S(Z|M)}{k}} \geq \lambda \exp{\frac{2S(X|M)}{k}} + | 1-\lambda | \exp{\frac{2S(Y|M)}{k}}\ .
  \label{eq:clcondepi}
\end{equation}

In quantum information theory, an analogous operation to the convolution \eqref{eq:conv} is given by the action of a beam splitter $U_\lambda$ with transmissivity $0\le\lambda\le1$ on a quantum state (i.e., a linear positive operator with unit trace) $\rho_{AB}$ which is bipartite on two $n$-mode Gaussian quantum systems $A,B$. This action has the form
\begin{equation}
  \rho_{AB} \mapsto \rho_{C} = \tr_2 \left(U_\lambda \rho_{AB} U_\lambda^\dagger \right)\ ,
\label{eq:bs}
\end{equation}
where $C$ is again an $n$-mode quantum system and $\tr_2$ denotes the partial trace over the second system.
The mathematical motivation of the study of this operation is that in the special case of a product state, that is $\rho_{AB}~=~\rho_A \otimes \rho_B$, it is formally similar to the convolution described in \eqref{eq:conv} on the level of Wigner functions. For the beam splitter \eqref{eq:bs}, several important inequalities in the same spirit as in classical information theory
have been established~\cite{KoeSmiEPI,depalmaepi14,MariPalma15,koenigconditionalepi2015,depalmaconditional}. For instance, the quantum entropy power inequality reads
\begin{equation}\label{eq:qEPI}
  \exp{\frac{S(C)}{n}} \geq \lambda \exp{\frac{S(A)}{n}} + (1-\lambda) \exp{\frac{S(B)}{n}}\ ,
\end{equation}
with $S(A) = S(\rho_A) = -\tr[\rho_A \log\rho_A]$ being the von Neumann entropy of a quantum state.
Unlike in the classical setting, a conditional entropy power inequality for the operation \eqref{eq:bs} does not trivially follow from the unconditioned inequality \eqref{eq:qEPI}. However, it was recently established in \cite{depalmaconditional} that such an inequality holds nonetheless: For a joint quantum state $\rho_{ABM}$ such that $A$ and $B$ are conditionally independent given the memory system $M$, we have

\begin{equation}
  \exp{\frac{S(C|M)}{n}} \geq \lambda \exp{\frac{S(A|M)}{n}} + (1-\lambda) \exp{\frac{S(B|M)}{n}}\ ,
  \label{eq:condepi}
\end{equation}
where $S(X|M) := S(XM) - S(M)$ is the quantum conditional entropy.
The conditional independence of $A$ and $B$ given $M$ is expressed with the condition that the quantum conditional mutual information equals zero:
\begin{equation}
I(A:B | M) := S(A|M) + S(B|M) - S(AB|M) = 0\ .
  \label{eq:condindependence}
\end{equation}

Our work concerns yet another convolution operation, which mixes a probability density function $f: \mathbb{R}^{2n} \rightarrow \mathbb{R}$ on phase space with an $n$-mode quantum state $\rho$,
\begin{equation}
  (f,\rho) \mapsto f \star \rho, \qquad \text{ where } \qquad f \star \rho = \int_{\mathbb{R}^{2n}} f(\xi) D(\xi) \rho D(\xi)^\dagger \frac{\text{d}^{2n}\xi}{(2\pi)^n}\ ,
  \label{eq:cqconv}
\end{equation}
where $D(\xi)$ are the Weyl displacement operators in phase space.
This operation was first introduced by Werner in~\cite{WernerHarmonicanalysis84}. Werner established a number of results regarding \eqref{eq:cqconv}, most notably a Young-type inequality. In~\cite{HubKoeVersh16}, more inequalities involving this operation were shown, most prominently the entropy power inequality
\begin{equation}
  \exp{\frac{S(f\star\rho)}{n}} \geq  \exp{\frac{S(f)}{n}} + \exp{\frac{S(\rho)}{n}}\ .
  \label{eq:cqepi}
\end{equation}
In the context of mixing times of semigroups, the authors in \cite{dattaetal16} have used this convolution extensively and proved various properties which are related to the discussion of the entropy power inequality.

\subsection{Our contribution}
Similarly to the work carried out in \cite{depalmaconditional} for the beam splitter, we prove the conditional version of the entropy power inequality for the convolution given by \eqref{eq:cqconv}.
Let us consider an $n$-mode Gaussian quantum system $A$, a generic quantum system $M$ and a classical system $R$ which ``stores'' a classical probability density function $\rho_R: \mathbb{R}^{2n} \rightarrow \mathbb{R}$.
Let us further consider the map $\cE: AR \rightarrow C$, $(\rho_A~\otimes~\rho_R)~\mapsto~\rho_R~\star~\rho_A$, linearly extended to generic states $\rho_{AR}$ as
\begin{equation}
  \rho_C = \cE(\rho_{AR}) = \int_{\mathbb{R}^{2n}}D(\xi)\,\rho_{A|R=\xi}\,{D(\xi)}^\dag\,\rho_R(\xi)\,\frac{\mathrm{d}^{2n}\xi}{(2\pi)^n}\;.
\end{equation}
We show in \autoref{thm:mainthm} that the conditional entropy of the output of $\cE\otimes\id_M: ARM \rightarrow CM$ is lower bounded as
\begin{equation}
  \exp{\frac{S(C|M)}{n}} \geq \exp{\frac{S(A|M)}{n}} + \exp{\frac{S(R|M)}{n}}\ ,
\end{equation}
if $I(A:R|M) = 0$, i.e., the systems $A$ and $R$ are conditionally independent given the system $M$.
As a special case, this inequality implies useful inequalities about the convolution~\eqref{eq:cqconv} in the case when $R$ is uncorrelated with $M$,
\begin{equation}
 \exp{\frac{S(C|M)}{n}} \geq \exp{\frac{S(A|M)}{n}} + \exp{\frac{S(\rho_R)}{n}}\ .
\end{equation}

In the particular case when $R$ is a Gaussian random variable with probability density function $f_{Z,t} = \exp\left(-\frac{\|\xi\|^2}{2t}\right)/t^n$, the inequality becomes
\begin{equation}
  \exp{\frac{S(C|M)}{n}} \geq \exp{\frac{S(A|M)}{n}} + et\ .
\end{equation}
The special cases mentioned above are important in various applications, as we will show later.

This conditional entropy power inequality is tight in the sense that it is saturated for any couple of values of $S(A|M)$ and $S(R|M)$ by an appropriate sequence of Gaussian input states, which we show in \autoref{thm:tightness}. This behaviour is similar to the case of the beam splitter.
On the way to this inequality, several intermediate results are proven which make up a set of information-theoretic inequalities regarding conditional Fisher information and conditional entropies. To complete the picture
of information-theoretic inequalities involving quantum conditional entropies, we apply our results to prove a number of additional inequalities in a spirit similar to the classical case. Among them there are the concavity of the quantum
conditional entropy along the heat flow (\autoref{thm:concavityheat}) and an isoperimetric inequality for quantum conditional entropies (Lemma \ref{lem:fisherisoperi}). Furthermore, we show in \autoref{sec:ea} how, similar to the case of the beam splitter, the conditional entropy power inequality implies a converse bound on the entanglement-assisted classical capacity of a non-Gaussian quantum channel, the classical noise channel defined in \eqref{eq:cqconv}.

Another part of our work regards the quantum Ornstein-Uhlenbeck (qOU) semigroup. It is the one-parameter semigroup of completely positive and trace-preserving (CPTP) maps $\left\{\cP^{(\mu,\lambda)}(t) = e^{t\cL_{\mu,\lambda}}\right\}_{t \geq 0}$ on the one-mode Gaussian quantum system $A$ generated by the Liouvillian
\begin{equation}
\cL_{\mu,\lambda} = \mu^2 \cL_- + \lambda^2 \cL_+ \qquad \text{ for } \mu > \lambda > 0\ ,
\end{equation}
where
\begin{equation}
\cL_+(\rho) = a^\dagger \rho a - \frac{1}{2} \{ aa^\dagger, \rho \} \qquad \text{ and } \qquad \cL_-(\rho) = a\rho a^\dagger - \frac{1}{2} \{ a^\dagger a, \rho\}\ ,
\end{equation}
and $a$ is the ladder operator of $A$.
This quantum dynamical semigroup has a unique fixed point given by
\begin{equation}
\omega^{\mu,\lambda} := \frac{\mu^2-\lambda^2}{\mu^2} \sum_{k=0}^\infty \left(\frac{\lambda^2}{\mu^2}\right)^k \proj{k}\ ,
\end{equation}
where $\{|k\rangle\}_{k\in\mathbb{N}}$ is the Fock basis of $A$.
It has been shown in~\cite{Carlen_2017} using methods of gradient flow that the quantum Ornstein-Uhlenbeck semigroup converges in relative entropy to the fixed point at an exponential rate given by the exponent $\mu^2-\lambda^2$,
\begin{equation}
D\left(\cP^{(\mu,\lambda)}(t)(\rho) \big\| \omega^{(\mu,\lambda)} \right) \leq e^{-(\mu^2-\lambda^2)t}D\left(\rho \big\| \omega^{(\mu,\lambda)}\right) \qquad \text{ for all } t \geq 0\ ,
\label{eq:qOUrate}
\end{equation}
where $D(\rho\|\sigma) = \tr\left[\rho \left(\log\rho - \log\sigma\right) \right]$ is the quantum relative entropy~\cite{holevobook}.

We show that a simple application of the linear version of the entropy power inequality \eqref{eq:qEPI} for the beam splitter is sufficient to prove this convergence rate. We also show a simple derivation of an analogous result for the case of a bipartite quantum system $AM$, where the system $A$ undergoes a qOU evolution, using the linear conditional entropy power inequality
for the beam splitter recently proven in \cite{depalmaconditional}. Specifically, we are going to show in \autoref{thm:qOUbipartite} that
\begin{equation}
  D\left((\cP^{(\mu,\lambda)}\otimes\id_M)(\rho_{AM})\big\| \omega_A^{(\mu,\lambda)}\otimes \rho_M \right) \leq e^{-(\mu^2-\lambda^2)t} D\left(\rho_{AM} \big\| \omega_A^{(\mu,\lambda)}\otimes \rho_M \right)\ ,
  \label{eq:bipartiterate}
\end{equation}
which directly implies the statement~\eqref{eq:qOUrate}. Finite-dimensional versions of the statement~\eqref{eq:bipartiterate} for general semigroups have recently been studied by Bardet~\cite{bardet17}. Our argument shows that entropy power inequalities are a useful tool to study the convergence rate of semigroups.

The proof of the unconditioned entropy power inequality \eqref{eq:cqepi} given in \cite{HubKoeVersh16} exhibits certain regularity issues regarding the Fisher information: the Fisher information was defined as the Hessian of a relative entropy, without a proof of well-definedness. Various proofs of the entropy power inequality for the beam splitter had similar issues~\cite{KoeSmiEPI,depalmaepi14,MariPalma15}. They were settled in~\cite{depalmaconditional} by the adoption of a proof technique which starts with an integral version of the quantum Fisher information. We adopt a similar approach here. Since the conditional entropy power inequality reduces to the unconditioned inequality in the case where the system $M$ is trivial, this also gives a more rigorous proof of the unconditioned entropy power inequality. As such,
our work can be seen as both a completion of the work carried out in~\cite{HubKoeVersh16} and a generalization thereof.

We now sketch the basic structure of the proof of our main result. The main ingredients in proving entropy power inequalities~\cite{Blachman65,KoeSmiEPI,MariPalma15,depalmaconditional,HubKoeVersh16} are similar in all proofs, which all use the evolution under the heat semigroup. These ingredients are the Fisher information, de Bruijn's identity, the Stam inequality, and a result on the asymptotic scaling of the entropy under the heat flow.
First we define a ``classical-quantum'' integral conditional Fisher information, by which we mean a Fisher information of a classical system which is conditioned on a quantum system. We show in \autoref{thm:integraldebruijn} that this quantity satisfies a de Bruijn identity, which links it to the change of the conditional entropy under the heat flow. We show the regularity of the integral conditional Fisher information in \autoref{thm:regularity} and then prove the conditional Stam inequality in \autoref{thm:stam}.
In the next part, we show in \autoref{thm:scaling} that the quantum conditional entropy of a classical system undergoing the classical heat flow evolution conditioned on a quantum system satisfies the same universal scaling which was shown for the quantum conditional entropy of a quantum system undergoing the quantum heat flow evolution conditioned on a quantum system. It is crucial for the proof of our conditional entropy power inequality that these two scalings are not only both universal but
also the same. This scaling then implies that asymptotically, the inequality we want to prove becomes an equality. Then it is left to show that it is enough to consider the inequality in the asymptotic limit, i.e., the difference of the two sides of the inequality behaves under the heat flow in a way which only makes the inequality ``worse''.

The paper is structured as follows: In \autoref{sec:prelim} we present bosonic quantum systems and the relevant quantities required for our discussion. In \autoref{sec:fisherinfo}, the integral version of the quantum conditional Fisher information is adapted to the convolution \eqref{eq:cqconv}. Sections \ref{sec:stam} and \ref{sec:scaling} are dedicated to the proof of various
inequalities that are central to the proof of entropy power inequalities, such as the Stam inequality and an asymptotic scaling of the conditional entropy. Section \ref{sec:epi} then proves the conditional entropy power inequality for the convolution \eqref{eq:cqconv} as our main result.
Optimality of the conditional entropy power inequality is shown in \autoref{sec:tightness}. This is followed by
the derivation of various related information-theoretic inequalities involving the quantum conditional entropy in \autoref{sec:apps}. Before concluding, we apply the conditional entropy power inequality to bound the convergence rate of bipartite systems where one system undergoes a quantum Ornstein-Uhlenbeck semigroup evolution in \autoref{sec:qOU}.

\section{Preliminaries}
\label{sec:prelim}
Let us consider an $n$-mode bosonic system~\cite{holevobook,serafinibook} with ``position'' and ``momentum'' operators $(Q_k, P_k)$, $k = 1, \dots, n$, for each mode which satisfy the canonical commutation relations $[Q_j, P_k] = i \delta_{j,k} \id$. If we denote the vector of position and momentum operators by $R = (Q_1, P_1, \dots, Q_n, P_n)$, the canonical commutation relations become
\begin{equation}
  [R_j, R_k] = i \Delta^{jk} \id, \qquad i,j = 1, \dots, 2n \ ,
  \label{eq:canonicalcomm}
\end{equation}
where $\Delta = \begin{pmatrix}0 & 1 \\ -1 & 0 \end{pmatrix}^{\oplus n}$ is the symplectic form.

The Weyl displacement operators are defined by
\begin{equation}
  D(\xi) := \exp\left(i \xi \cdot (\Delta^{-1} R) \right), \qquad \text{ for } \xi \in \mathbb{R}^{2n}\ .
  \label{eq:displacement}
\end{equation}
The displacement operators satisfy the commutation relations
\begin{equation}
  D(\xi) D(\eta) = \exp\left( -\frac{i}{2} \xi \cdot (\Delta^{-1} \eta)\right) D(\xi + \eta), \qquad \text{ for } \xi, \eta \in \mathbb{R}^{2n}\ ,
  \label{eq:displacementcomm}
\end{equation}
and the ``displacement property'' on the mode operators
\begin{equation}
D(\xi)^\dagger R_j D(\xi) = R_j + \xi_j \id\ .
\end{equation}

Given an $n$-mode quantum state $\rho$, we define its first moments as
\begin{equation}
  d_k(\rho) := \tr[ R_k \rho ], \qquad \text{ for } k = 1, \dots, 2n\ ,
  \label{eq:firstmoments}
\end{equation}
and its covariance matrix (for finite first moments) as
\begin{equation}
  \Gamma_{kl}(\rho) := \frac{1}{2} \tr\left[\left\{R_k - d_k(\rho), R_l - d_l(\rho) \right\} \rho \right], \qquad k,l = 1, \dots, 2n\ ,
  \label{eq:covmat}
\end{equation}
with the anticommutator $\{X,Y\} := XY - YX$.

The aforementioned concepts of displacements and first and second moments are the quantum analogs of the classical concepts. For a probability distribution function $f: \mathbb{R}^{2n} \rightarrow \mathbb{R}$, we
define its displacement by a vector $\eta \in \mathbb{R}^{2n}$ as
\begin{equation}
f^{(\eta)}(\xi) = f(\xi - \eta)\ .
\label{eq:cldisplacement}
\end{equation}

Furthermore, we denote the energy of the function $f$ by the sum of its second moments,
\begin{equation}
  E(f) = \sum_{k=1}^{2n} \int_{\mathbb{R}^{2n}} \xi_k^2 f(\xi) \frac{\mathrm{d}^{2n}\xi}{(2\pi)^n}\ .
\end{equation}
The quantities $\mu_k = \int_{\mathbb{R}^{2n}} \xi_k f(\xi) \frac{\mathrm{d}^{2n}\xi}{(2\pi)^n}$ are called the first moments of $f$, and
\begin{equation}
  \gamma_{kl} = \int_{\mathbb{R}^{2n}} f(\xi) (\xi_k-\mu_k)(\xi_l-\mu_l) \frac{\mathrm{d}^{2n} \xi}{(2\pi)^n}
\end{equation}
is called the covariance matrix of $f$. We remark that we have rescaled the Lebesgue measure on $\mathbb{R}^{2n}$ in these definitions, which we have done purely for convenience.

\begin{definition}[Quantum heat semigroup] The quantum heat semigroup is the following time evolution for any quantum state $\rho$:
\begin{align}
  \cN(t) (\rho) &:= \int_{\mathbb{R}^{2n}} e^{-\frac{\|\xi\|^2}{2t}} \rho^{(\xi)} \frac{\mathrm{d}^{2n}\xi}{(2\pi t)^n} \qquad \mathrm{ for }\;  t > 0\ ,\\
\cN(0) &:= \id\ ,
\end{align}
where $\rho^{(\xi)} = D(\xi) \rho D(\xi)^\dagger$ is a displacement of the state $\rho$ by $\xi\in\mathbb{R}^{2n}$.

The quantum heat semigroup has a semigroup structure, that is, for any $s,t \geq 0$, we have
\begin{equation}
\cN(s) \circ \cN(t) = \cN(s+t)\ .
\end{equation}
\end{definition}

We note that if $f_{Z,t}(\xi) = \exp\left(-\frac{\|\xi\|^2}{2t}\right)/t^n$ is the probability distribution of a Gaussian random variable with covariance matrix $t \id_{2n}$, then  we have
\begin{equation}
\cN(t)(\rho) = f_{Z,t} \star \rho\ .
\end{equation}

The quantum heat semigroup is the quantum analog of the classical heat semigroup, which we will repeat here. It can be written in an analogous way to the quantum heat semigroup:

\begin{definition}[Classical heat semigroup] The classical heat semigroup is the following time evolution defined on a function $f: \mathbb{R}^{2n} \rightarrow \mathbb{R}$:
\begin{align}
  \left(\cN_{\mathrm{cl}}(t)(f)\right)(\eta) &:= \int_{\mathbb{R}^{2n}} e^{-\frac{\| \xi \|^2}{2t}}f^{(\xi)}(\eta) \frac{\mathrm{d}^{2n}\xi}{(2\pi t)^{n}}\ , \\
  \cN_{\mathrm{cl}}(0) &:= \id \ .
\end{align}

We also have that for any $s,t \geq 0$,
\begin{equation}
  \cN_{\mathrm{cl}}(s) \circ \cN_{\mathrm{cl}}(t) = \cN_{\mathrm{cl}}(s+t) \ .
\end{equation}
\end{definition}

We note again that we have
\begin{equation}
  \cN_{\mathrm{cl}}(t)(f) = f_{Z,t} \star f\ ,
\end{equation}
where
\begin{equation}
  (g \star f)(\eta) := \int_{\mathbb{R}^{2n}} g(\xi) f(\eta-\xi) \frac{\text{d}^{2n}\xi}{(2\pi)^n}
\end{equation}
is the well-known classical convolution of the two functions $g$ and $f$ (with a factor of $(2\pi)^n$ in the Lebesgue measure on $\mathbb{R}^{2n}$ which we introduce purely for convenience).

The convolution \eqref{eq:cqconv} is compatible with displacements and with the heat semigroup evolution in a convenient way, which is stated in the following two lemmas:

\begin{lemma}[Compatibility with displacements of the convolution \eqref{eq:cqconv}]\cite[Lemma 2]{HubKoeVersh16}
Let $f: \mathbb{R}^{2n} \rightarrow \mathbb{R}$ be a probability distribution and $\rho$ an $n$-mode quantum state. Then we have for any $\xi_1,\xi_2 \in\mathbb{R}^{2n}$,
\begin{equation}
  \left(f \star \rho\right)^{(\xi_1 + \xi_2)} = f^{(\xi_1)} \star \rho^{\left(\xi_2\right)}\ ,
\end{equation}
where $\rho^{(\xi)} = D(\xi) \rho D(\xi)^\dagger$.
\label{lem:compdisp}
\end{lemma}
\begin{remark} Lemma 2 in~\cite{HubKoeVersh16} only states the compatibility for the case where $\xi_1, \xi_2$ are parallel. Nonetheless, the proof given there also works to prove the statement above.
\end{remark}

\begin{lemma}[Compatibility with the heat semigroup of the convolution \eqref{eq:cqconv}]\cite[Lemma 5]{HubKoeVersh16}
Assume the same prerequisites as in Lemma \ref{lem:compdisp} and let $t_1, t_2 \geq 0$. Then we have
\begin{equation}
\cN(t_1+t_2) \left(f \star \rho\right) = \cN_{\mathrm{cl}}(t_1)(f) \star \cN(t_2)(\rho)\ .
\end{equation}
\end{lemma}

\begin{definition}[Shannon differential entropy]
  For a classical $\mathbb{R}^{2n}$-valued random variable $X$ with a probability density function $f: \mathbb{R}^{2n} \rightarrow \mathbb{R}$, we define the Shannon differential entropy as
  \begin{equation}
    S(X) = S(f) = -\int_{\mathbb{R}^{2n}} f(\xi) \log f(\xi) \frac{\mathrm{d}^{2n}\xi}{(2\pi)^n}\ .
  \end{equation}
\end{definition}
We continue with a short review of Gaussian quantum states.
An $n$-mode quantum state $\rho_G$ is called Gaussian if it has the following form \cite{holevobook}:
\begin{equation}
\rho_G = \frac{\exp\left[-\frac{1}{2} \sum_{k,l =1}^{2n} \left(R_k - d_k\right) h_{kl} \left(R_l - d_l\right) \right]  }{\tr \exp\left[-\frac{1}{2}\sum_{k,l=1}^{2n}\left(R_k - d_k\right) h_{kl} \left(R_l - d_l\right) \right] }\ ,
\end{equation}
where $h$ is a positive definite real $2n \times 2n$ matrix and $d \in \mathbb{R}^{2n}$ is the vector of first moments of the state. The entropy of such a Gaussian state is given by
\begin{equation}
S(\rho_G) = \sum_{k=1}^n g\left(\nu_k - \frac{1}{2}\right)\ ,
\end{equation}
where $g(N) := (N+1)\log(N+1) - N\log N$ and $\nu_1, \dots, \nu_n$ are the symplectic eigenvalues of the covariance matrix $\Gamma = \frac{\Delta}{2}\left(\tan\frac{h\,\Delta}{2}\right)^{-1}$, i.e., the absolute values of the eigenvalues of $\Delta^{-1} \Gamma$.

A Gaussian state is called thermal if its first moments are zero and the matrix $h$ is proportional to the identity. Such thermal states have the special form
\begin{equation}
\omega_\beta = \frac{e^{-\beta H}}{\tr e^{-\beta H}}, \qquad h = \beta \id_{2n}, \qquad \beta > 0\
\end{equation}
for the Hamiltonian of $n$ harmonic oscillators $H = \frac{1}{2} \sum_{k=1}^{2n} R_k^2 - \frac{n}{2} \id$. Gaussian states fulfill a special extremality property. Among all states $\rho$ with a given average energy $\tr\left[H \rho\right]$, thermal states maximize the von Neumann entropy. Furthermore, among all states with fixed covariance matrix, the Gaussian state is the one with maximal entropy~\cite{Holevo_1999,Wolfetal06}.

In our proofs, we are going to require the notion of quantum conditional Fisher information of quantum systems which was introduced in~\cite{depalmaconditional}. We repeat the main properties of this quantity here. For a thorough definition and proofs we refer to \cite{depalmaconditional}.
Before giving this definition, we clarify the notion of ``classical-quantum'' states on a system $RM$ if the classical system $R$ is continuous.
A state $\rho_{RM}$ on $RM$ is a probability measure on $R$ which takes values in the trace class operators, i.e., a measurable collection of trace class operators on $M$ $\{\rho_{MR}(\xi)\}_{\xi\in\mathbb{R}^{2n}}$ with the normalization condition
\begin{equation}
  \int_{\mathbb{R}^{2n}} \tr_M[\rho_{MR}(\xi)] \frac{\mathrm{d}^{2n}\xi}{(2\pi)^n} = 1\ .
\end{equation}
This state ``stores'' a classical probability distribution $\rho_R$ in the classical system $R$ if its marginal on $R$ has $\rho_R$ as probability distribution.
The marginals of $\rho_{MR}$ are
\begin{equation}
  \rho_M = \int_{\mathbb{R}^{2n}}\rho_{MR}(\xi) \frac{\mathrm{d}^{2n}\xi}{(2\pi)^n},\qquad \rho_R(\xi) = \tr_M[\rho_{MR}(\xi)],
\end{equation}
and the conditional states on $M$ given the value of $\xi$ are
\begin{equation}
\rho_{M|R=\xi} = \frac{\rho_{MR}(\xi)}{\rho_R(\xi)}.
\end{equation}
We do not consider the case where the probability measure $\rho_R$ is not absolutely continuous with respect to the Lebesgue measure, since in this case its Shannon differential entropy is not defined.
For a more detailed discussion, we refer to \cite[Section III.A.3]{Furrer_2014} and references therein (\cite{Murphy1990173} and \cite[Chapter 4.6-4.7]{takesakibook}).

We can also define displacements of such a classical-quantum state: We write $\rho_{RM}^{(x, y)}$ to denote a state where the classical system $R$ has been displaced by $x \in \mathbb{R}^{2n}$ and the quantum system $M$ has been displaced by $y \in \mathbb{R}^{2n}$.

\begin{definition}[Quantum integral conditional Fisher information]\cite[Definition 6]{depalmaconditional}
  Let $A$ be an $n$-mode bosonic quantum system, and $M$ a generic quantum system. Let $\rho_{AM}$ be a quantum state on $AM$. For any $t \geq 0$, the integral Fisher information of $A$ conditioned on $M$ is given by
  \begin{align}
    \Delta_{A|M}(\rho_{AM})(t) &:= I(A:Z|M)_{\sigma_{AMZ}(t)} \geq 0\ , \qquad t > 0\ ,\\
      \Delta_{A|M}(\rho_{AM})(0) &:= 0\ ,
  \end{align}
  where $Z$ is a classical Gaussian random variable with values in $\mathbb{R}^{2n}$ and probability density function
  \begin{equation}
    f_{Z,t}(z) = \frac{e^{-\frac{|z|^2}{2t}}}{t^n},\qquad z \in \mathbb{R}^{2n}\ ,
  \end{equation}
  and $\sigma_{AMZ}(t)$ is the quantum state on $AMZ$ such that its marginal on $Z$ is $f_{Z,t}$ and for any $z \in \mathbb{R}^{2n}$,
  \begin{equation}
    \sigma_{AM|Z=z}(t) = D_A(z) \rho_{AM} D_A(z)^\dagger\ .
  \end{equation}
\end{definition}
\begin{definition}[Quantum conditional Fisher information]\cite[Definition 7, Proposition 1]{depalmaconditional}
  Let $\rho_{AM}$ be a quantum state on $AM$ such that the marginal $\rho_{A}$ has finite energy and the marginal $\rho_M$ has finite entropy. Then we define the quantum conditional Fisher information of $A$ conditioned on $M$ as
  \begin{equation}
    J(A|M)_{\rho_{AM}} := \lim_{t\rightarrow 0} \frac{\Delta_{A|M}(\rho_{AM})(t)}{t} = \frac{\mathrm{d}}{\mathrm{d}t} S(A|M)_{(\cN_A(t)\otimes\id_M)(\rho_{AM})}\bigg|_{t=0}\ .
\end{equation}
As shown in~\cite{depalmaconditional}, this limit always exists.
\end{definition}

Finally, we are going to require a notion of conditional entropy of a classical system which is conditioned on a quantum system. If the system on which we condition is classical, the conditional entropy is simply
\begin{equation}
S(A|M) = \int_M S(A|M = m) \mathrm{d}p_M(m)\ ,
\end{equation}
where $p_M$ is the probability distribution of $M$. This definition is independent of whether the system $A$ is classical or quantum. We now define the conditional entropy of a classical system which is conditioned on a quantum system
in a way such that the chain rule for entropies is preserved.

\begin{definition}[Quantum conditional entropy of classical-quantum systems]
Let $R$ be a classical system, $M$ a quantum system. We define the conditional entropy of $R$ given $M$ as
\begin{equation}
S(R|M) = S(M|R) + S(R) - S(M)\ ,
\end{equation}
whenever the three quantities appearing on the right handside are finite.
\end{definition}
The case where $S(M|R), S(R)$, and $S(M)$ are not finite will not be part of our consideration.

\section{Quantum integral conditional Fisher information}
\label{sec:fisherinfo}
In this section we consider a generic quantum system $M$ and a classical system $R$.
We are going to define the quantum integral conditional Fisher information of $R$ conditioned on $M$ and prove a de Bruijn identity as well as a number of useful properties.

\begin{definition}[quantum integral conditional Fisher information]
  For a quantum state $\rho_{RM}$ on $RM$ whose marginal on $R$ is $\rho_R: \mathbb{R}^{2n} \rightarrow \mathbb{R}$ and $t \geq 0$, define the integral Fisher information of $R$ conditioned on $M$ as
\begin{align}
\Delta_{R|M}(\rho_{RM})(t) &:= I(R:Z | M)_{\sigma_{RZM}(t)}\ ,\\
\Delta_{R|M}(\rho_{RM})(0) &:= 0\ ,
\end{align}
where $Z$ is a classical Gaussian random variable with probability density function equal to
\begin{equation}
f_{Z,t}(\xi) = \frac{e^{-\frac{|\xi|^2}{2t}}}{t^n}, \qquad \xi \in \mathbb{R}^{2n}\ ,
\end{equation}
and $\sigma_{RZM}(t)$ is the quantum state on $RZM$ such that its marginal on $Z$ is equal to $f_{Z,t}$, and for any $z \in \mathbb{R}^{2n}$, we have
\begin{equation}
  \sigma_{RM|Z=z}(t) = \rho_{RM}^{(z,0)} \ .
\end{equation}
The marginal of $\sigma_{RZM}(t)$ on $RM$ is equal to
\begin{equation}
\sigma_{RM}(t) = \left( \cN_{\mathrm{cl}}(t) \otimes \id_M\right) (\rho_{RM})\ .
\end{equation}
The marginal on $R$ has probability density function $\cN_{\mathrm{cl}}(t)(\rho_R)$.
\end{definition}

\begin{theorem}[Integral conditional de Bruijn identity]
\label{thm:integraldebruijn}
\begin{equation}
  \Delta_{R|M}(\rho_{RM})(t) = S(R|M)_{(\cN_{\mathrm{cl}}(t)\otimes \id_{M})(\rho_{RM})} - S(R|M)_{\rho_{RM}}\ .
\end{equation}
\end{theorem}
\begin{proof}
We use the definition of the conditional mutual information as well as the definition of the conditional quantum entropy when the system on which we condition is classical. We calculate
\begin{align}
I(R:Z | M)_{\sigma_{RMZ}} &= S(R|M)_{\sigma_{RMZ}} - S(R|MZ)_{\sigma_{RMZ}}\\
&= S(R|M)_{\sigma_{RM}} - \int_{\mathbb{R}^{2n}} S(R|M)_{\sigma_{RM|Z = z}} f_{Z,t}(z)\,\frac{\mathrm{d}^{2n}z}{(2\pi)^n} \\
&= S(R|M)_{\sigma_{RM}} - \int_{\mathbb{R}^{2n}} S(R|M)_{\rho_{RM}} f_{Z,t}(z)\,\frac{\mathrm{d}^{2n}z}{(2\pi)^n}\\
&= S(R|M)_{\sigma_{RM}} - S(R|M)_{\rho_{RM}}\ .
\end{align}
The second to last step follows because the entropy is invariant under displacements of the classical system.
\end{proof}
We now show that the integral conditional Fisher information defined as above, as a function of $t$, is continuous, increasing, and concave. The proof strategy is similar to the proof of regularity for the quantum integral conditional Fisher information given in~\cite{depalmaconditional}.

\begin{lemma}[Continuity of the integral conditional Fisher information]
  \label{lem:continuity}
Let $\rho_{RM}$ be a state such that the function $\mathbb{R}^{2n}\ni\xi\mapsto\rho_{M|R=\xi}$ is continuous with respect to the trace norm and the marginal $\rho_R$ has finite average energy.
Then, the function $t\mapsto \Delta_{R|M}(\rho_{RM})(t)$ is continuous for any $t\ge0$.
\end{lemma}
\begin{proof}
From the de Bruijn identity \autoref{thm:integraldebruijn}, it is sufficient to prove that
\begin{equation}
\lim_{t\to 0}S(R|M)(\rho_{RM}(t)) = S(R|M)(\rho_{RM})\;,
\end{equation}
where we have defined for any $t\ge0$,
\begin{equation}
\rho_{RM}(t) = (\cN_{\text{cl}}(t)\otimes \id_M)(\rho_{RM})\;.
\end{equation}
From the data processing inequality, for any $t\ge0$
\begin{equation}\label{eq:dpRM}
S(R|M)(\rho_{RM}(t)) \ge S(R|M)(\rho_{RM})\;.
\end{equation}
It is then sufficient to prove that
\begin{equation}
\limsup_{t\to 0}S(R|M)(\rho_{RM}(t)) \le S(R|M)(\rho_{RM})\;.
\end{equation}
We have from the chain rule
\begin{equation}
S(R|M)((\cN_{\text{cl}}(t)\otimes \id_M)(\rho_{RM})) = S(M|R)(\rho_{RM}(t)) + S(\rho_R(t)) - S(\rho_M)\;.
\end{equation}
From \cite{ambrosio2008gradient}, Remark 9.3.8, and \cite{ambrosio2000functions,buttazzo1989semicontinuity,goffman1964sublinear}, the Shannon differential entropy is upper semicontinuous on the set of probability measures on $\mathbb{R}^{2n}$ absolutely continuous with respect to the Lebesgue measure and with finite average energy, and
\begin{equation}
\limsup_{t\to0}S(\rho_R(t))\le S(\rho_R)\;.
\end{equation}
On the other hand, we have
\begin{equation}\label{eq:int}
S(\rho_M) - S(M|R)(\rho_{RM}(t)) = \int_{\mathbb{R}^{2n}}D(\rho_{M|R=\xi}(t)\|\rho_M)\,\rho_R(t)(\xi)\,\frac{\mathrm{d}^{2n}\xi}{(2\pi)^n}\;.
\end{equation}
Since the function $t\mapsto\rho_{M|R=\xi}(t)$ is continuous with respect to the trace norm, we have for any $\xi\in\mathbb{R}^{2n}$
\begin{equation}
\lim_{t\to0}\|\rho_{M|R=\xi}(t) - \rho_{M|R=\xi}\|_1=0\;.
\end{equation}
Because the relative entropy is positive, we get from Fatou's lemma
\begin{align}\label{eq:fatou}
&\int_{\mathbb{R}^{2n}}\liminf_{t\to 0}D(\rho_{M|R=\xi}(t)\|\rho_M)\,\rho_R(t)(\xi)\,\frac{\mathrm{d}^{2n}\xi}{(2\pi)^n}\nonumber\\
& \le \liminf_{t\to 0}\int_{\mathbb{R}^{2n}}D(\rho_{M|R=\xi}(t)\|\rho_M)\,\rho_R(t)(\xi)\,\frac{\mathrm{d}^{2n}\xi}{(2\pi)^n}\;.
\end{align}
Since the relative entropy is lower semicontinuous, we have for any $\xi\in\mathbb{R}^{2n}$
\begin{equation}\label{eq:sc}
D(\rho_{M|R=\xi}\|\rho_M) \le \liminf_{t\to 0}D(\rho_{M|R=\xi}(t)\|\rho_M)\;.
\end{equation}
Combining \eqref{eq:fatou}, \eqref{eq:sc}, and \eqref{eq:int}, we get
\begin{equation}
\limsup_{t\to0}S(M|R)(\rho_{RM}(t)) \le S(M|R)(\rho_{RM})\;.
\end{equation}
\end{proof}

\begin{lemma} For any $s,t \geq 0$,
\begin{equation}
  \Delta_{R|M}\left((\cN_{\mathrm{cl}}(s) \otimes \id_M)(\rho_{RM})\right)(t) = I(R:Z|M)_{(\cN_{\mathrm{cl}}(s) \otimes \id_{MZ})(\sigma_{RMZ}(t)) }\ .
\end{equation}
\end{lemma}
\begin{proof}
Follows from the semigroup structure of $\cN_{\mathrm{cl}}$.
\end{proof}

\begin{lemma}
\label{lem:dp}
For any $s,t \geq 0$,
\begin{equation}
  \Delta_{R|M}\left( (\cN_{\mathrm{cl}}(s) \otimes \id_M)(\rho_{RM} \right)(t) \leq \Delta_{R|M}(\rho_{RM})(t)\ .
\end{equation}
\end{lemma}
\begin{proof}
Follows from the data processing inequality for the quantum mutual information.
\end{proof}

\begin{lemma}
\label{lem:increasing}
For any $s,t \geq 0$,
\begin{align}
\Delta_{R|M}(\rho_{RM})(s+t) &= \Delta_{R|M}(\rho_{RM})(s) + \Delta_{R|M}\left((\cN_{\mathrm{cl}}(s) \otimes \id_M)(\rho_{RM} \right)(t)\\
&\geq \Delta_{R|M}(\rho_{RM})(s)\ .
\end{align}
\end{lemma}
\begin{proof}
Follows from \autoref{thm:integraldebruijn}.
\end{proof}

\begin{theorem}[Regularity of the integral conditional Fisher information]
\label{thm:regularity}
For any quantum state $\rho_{RM}$ on $RM$ such that the conditions of Lemma \ref{lem:continuity} are fulfilled,
the integral conditional Fisher information $\Delta_{R|M}(\rho_{RM})(t)$ is a continuous, increasing, and concave function of $t$.
\end{theorem}
\begin{proof}
Continuity was shown in Lemma \ref{lem:continuity} and the fact that the conditional Fisher information is increasing follows from Lemma \ref{lem:increasing}.

For concavity, by continuity it is enough to prove that for $0~\leq~s~\leq~t$, we have
\begin{equation}
\Delta_{R|M}(\rho_{RM})\left(\frac{s+t}{2} \right) \geq \frac{\Delta_{R|M}(\rho_{RM})(s) + \Delta_{R|M}(\rho_{RM})(t)}{2}\ .
\end{equation}

This can be written as
\begin{align}
&\Delta_{R|M}(\rho_{RM})\left( \frac{s+t}{2} \right) - \Delta_{R|M}(\rho_{RM})(s)\nonumber\\
& \geq \Delta_{R|M}(\rho_{RM})(t) - \Delta_{R|M}(\rho_{RM})\left(\frac{s+t}{2}\right)\ .
\label{eq:concav}
\end{align}
By Lemma \ref{lem:increasing}, this can be restated as
\begin{equation}
\Delta_{R|M}(\rho_{RM}(s))\left( \frac{t-s}{2}\right) \geq \Delta_{R|M}\left(\left(\cN_{\mathrm{cl}}\left(\frac{t-s}{2}\right)\otimes \id_M \right)(\rho_{RM}(s))\right)\left(\frac{t-s}{2}\right)\ ,
\end{equation}
for $\rho_{RM}(s) := \left(\cN_{\mathrm{cl}}(s)\otimes \id_M\right)(\rho_{RM})$.
But this holds because of Lemma \ref{lem:dp}.
\end{proof}

\section{Quantum conditional Fisher information}
\label{sec:stam}
\begin{definition}
For a quantum state $\rho_{RM}$ on $RM$ such that the conditions of Lemma \ref{lem:continuity} are fulfilled, we define the Fisher information of $R$ conditioned on $M$ as
\begin{equation}
  J(R|M)_{\rho_{RM}} := \lim_{t\rightarrow 0} \frac{\Delta_{R|M}(\rho_{RM})(t)}{t}\ .
\end{equation}
\end{definition}
This limit always exists because the function $t \mapsto \Delta_{R|M}(\rho_{RM})(t)$ is continuous and concave by \autoref{thm:regularity}.

\begin{proposition}[Quantum conditional de Bruijn] Assume the hypotheses of \autoref{thm:regularity}. Then we have
\begin{equation}
  J(R|M)_{\rho_{RM}} = \frac{\mathrm{d}}{\mathrm{d}t} S(R|M)_{(\cN_{cl}(t) \otimes \id_M)(\rho_{RM})}\bigg|_{t = 0}\ .
  \label{eq:debruijn}
\end{equation}
\end{proposition}
\begin{proof}
Follows from the integral conditional de Bruijn identity given in \autoref{thm:integraldebruijn}.
\end{proof}
\subsection{Stam inequality}
\begin{theorem}
\label{thm:stam}
  Let $A$ be an $n$-mode quantum system, $R$ be a classical system and $M$ be a generic quantum system. Let $\rho_{ARM}$ be a quantum state on $ARM$ such that its marginal on $R$ has a probability density function $\rho_{R}: \mathbb{R}^{2n} \rightarrow \mathbb{R}$. Let $\rho_{ARM}$ further fulfill
\begin{equation}
  \tr[H\rho_A] < \infty, \qquad E(\rho_R) < \infty,\qquad S(\rho_M) < \infty\ .
\end{equation}
Let us suppose that $A$ and $R$ are conditionally independent given $M$,
\begin{equation}
I(A:R|M)_{\rho_{ARM}} = 0\ .
\end{equation}
Then the linear conditional Stam inequality holds,
\begin{equation}
  J(C|M)_{\rho_{CM}} \leq \lambda^2 J(A|M)_{\rho_{AM}} + (1-\lambda)^2 J(R|M)_{\rho_{RM}}\qquad \forall \lambda \in [0,1]\ ,
\label{eq:linearstam}
\end{equation}
where
\begin{equation}
  \rho_{CM} := (\cE\otimes \id_M)(\rho_{ARM}) = \int_{\mathbb{R}^{2n}}D(\xi)\,\rho_{AM|R=\xi}\,{D(\xi)}^\dag\,\rho_R(\xi)\,\frac{\mathrm{d}^{2n}\xi}{(2\pi)^n}\;.
\end{equation}
Choosing $\lambda = \frac{J(R|M)}{J(R|M) + J(A|M)}$, we obtain the conditional Stam inequality
\begin{equation}
  \frac{1}{J(C|M)_{\rho_{CM}}} \geq \frac{1}{J(A|M)_{\rho_{AM}}} + \frac{1}{J(R|M)_{\rho_{RM}}}\ .
\label{eq:stam}
\end{equation}
\end{theorem}
\begin{proof}
We prove the following:
\begin{equation}
\Delta_{C|M}(\rho_{CM})(t) \leq \Delta_{A|M}(\rho_{AM})(\lambda^2 t) + \Delta_{R|M}(\rho_{RM})((1-\lambda)^2t)\ .
\end{equation}
Because $\Delta$ is increasing and concave the Stam inequality follows taking the derivative at $t = 0$.

By definition, we have for any $t \geq 0$ that
\begin{equation}
  \Delta_{C|M}(\rho_{CM})(t) = I(C:Z|M)_{\sigma_{CMZ}(t)}\ ,
\end{equation}
for an $\mathbb{R}^{2n}$-valued Gaussian random variable $Z$ with probability density function
\begin{equation}
  f_{Z,t}(z) = \frac{e^{-\frac{\| z \|^2}{2t}}}{t^n}, \qquad z \in \mathbb{R}^{2n}\ ,
\end{equation}
and $\sigma_{CMZ}(t)$ has $f_{Z,t}$ as marginal on $Z$ and for any $z \in \mathbb{R}^{2n}$, it fulfills
\begin{equation}
  \sigma_{CM|Z=z}(t) = D_C(z) \rho_{CM} D_C(z)^\dagger\ .
\end{equation}
We now define the state $\sigma_{ARMZ}(t)$ as the state with marginal on $Z$ equal to $f_{Z,t}$ and for any $z\in\mathbb{R}^{2n}$,
\begin{equation}
  \sigma_{ARM|Z=z} = \rho_{ARM}^{(\lambda z, (1-\lambda)z)}\ ,
\end{equation}
i.e., the system $A$ is displaced by $\lambda z$ and the system $R$ is displaced by $(1-\lambda)z$.
By compatibility of the convolution \eqref{eq:cqconv} with displacements, we have
\begin{equation}
  \sigma_{CMZ}(t) = (\cE\otimes\id_{MZ})(\sigma_{ARMZ}(t))\ .
\end{equation}
We notice that
\begin{align}
  I(A:R|MZ)_{\sigma_{ARMZ}} &= \int_{\mathbb{R}^{2n}} I(A:R|M)_{\sigma_{ARM|Z=z}}\, f_{Z,t}(z)\,\frac{\mathrm{d}^{2n}z}{(2\pi)^n} \\
  &= \int_{\mathbb{R}^{2n}} I(A:R|M)_{\rho_{ARM}} \,f_{Z,t}(z)\,\frac{\mathrm{d}^{2n}z}{(2\pi)^n} = 0\ .
\end{align}

Now we obtain by data processing
\begin{align}
I(C:Z|M)(t) &\leq I(AR:Z|M)(t) \\
&= I(A:Z|M)(t) + I(R:Z|M)(t) + I(A:R|MZ)(t) - I(A:R|M)(t)\\
&\leq I(A:Z|M)(t) + I(R:Z|M)(t)\ .
\end{align}
The last inequality follows because $I(A:R|M)(t) \geq 0$.
In analogy to~\cite[Eqs. (79)-(81)]{depalmaconditional}, we can show that
\begin{align}
  I(A:Z|M)_{\sigma_{AMZ}(t)} &= \Delta_{A|M}(\rho_{AM})\left(\lambda^2 t\right)\ ,\\
  I(R:Z|M)_{\sigma_{RMZ}(t)} &= \Delta_{R|M}(\rho_{RM})\left( (1-\lambda)^2 t\right)\ .
\end{align}
This follows from the definition of $\sigma_{ARMZ}$ and the integral conditional de Bruijn identities. The claim follows.
\end{proof}

\section{Universal scaling}
\label{sec:scaling}
\begin{theorem}
  \label{thm:scaling}
Let $R$ be a classical system and $M$ be a quantum system. Let $\rho_{RM}$ be a quantum state on $RM$ such that its marginals have finite entropies. Then we have
\begin{equation}
\lim_{t \rightarrow \infty} (S(R|M)_{(\cN_{\mathrm{cl}}(t)\otimes \id_M)(\rho_{RM})} - n \log t - n) = 0.
\end{equation}
\begin{proof}
{\bf Upper bound.} We have
\begin{align}
  S(R|M)_{(\cN_{\mathrm{cl}}(t)\otimes \id_M)(\rho_{RM})} \leq S(R)_{\cN_{\mathrm{cl}}(t)(\rho_R)}\ .
 \end{align}
 We know from the analysis of the classical heat flow \cite{Blachman65} that the right hand side scales as $n \log t + n$. \\

{\bf Lower bound.}
By concavity, we can restrict to pure $\rho_{RM}$.
The pure states of the classical-quantum system $RM$ are the tensor product of a Dirac delta on $R$ with a pure state on $M$, hence $R$ and $M$ are independent and
\begin{equation}
  S(R|M)_{(\cN_{\mathrm{cl}}(t)\otimes \id_M)(\rho_{RM})} = S(R)_{\cN_{\mathrm{cl}}(t)(\rho_R)}\ .
\end{equation}
Finally, the scaling of the classical entropy $S(R)_{\cN_{\mathrm{cl}}(t)(\rho_R)}$ is known to be equal to $n \log t + n$ from \cite{Blachman65}, which concludes the proof.
\end{proof}
\end{theorem}
\section{Entropy power inequality}
\label{sec:epi}
\begin{theorem}[Conditional entropy power inequality for the convolution \eqref{eq:cqconv}]
\label{thm:mainthm}
  Let $A$ be an $n$-mode quantum system, $R$ a classical system and $M$ a generic quantum system. Let $\rho_{ARM}$ be a quantum state on $ARM$ such that its marginal on $R$ has a probability density function $\rho_R: \mathbb{R}^{2n} \rightarrow \mathbb{R}$. Let $\rho_{ARM}$ further fulfill
  \begin{equation}
  \tr[H\rho_A] < \infty, \qquad E(\rho_R) < \infty,\qquad S(\rho_M) < \infty\ .
\end{equation}
Let us suppose that $A$ and $R$ are conditionally independent given $M$:
\begin{equation}
I(A:R|M)_{\rho_{ARM}} = 0\ ,
\end{equation}
and let
\begin{equation}
  \rho_{CM} := (\cE\otimes \id_M)(\rho_{ARM}) = \int_{\mathbb{R}^{2n}}D(\xi)\,\rho_{AM|R=\xi}\,{D(\xi)}^\dag\,\rho_R(\xi)\,\frac{\mathrm{d}^{2n}\xi}{(2\pi)^n}\;.
\end{equation}
Then, for any $0 \leq \lambda \leq 1$ the linear conditional entropy power inequality holds:
\begin{equation}
 \frac{S(C|M)}{n} \geq \lambda \frac{S(A|M)}{n} + (1-\lambda) \frac{S(R|M)}{n} - \lambda \log \lambda - (1-\lambda) \log (1-\lambda)\ .
\end{equation}
Optimizing over $\lambda$ and choosing $\lambda = \frac{e^{S(A|M)/n}}{e^{S(A|M)/n} + e^{S(R|M)/n}}$, we obtain the conditional entropy power inequality for the convolution \eqref{eq:cqconv},
\begin{equation}
  \boxed{
    \exp{\frac{S(C|M)}{n}} \geq \exp{\frac{S(A|M)}{n}} + \exp{\frac{S(R|M)}{n}}\ .}
  \label{eq:cqcondepi}
\end{equation}
In particular, if the classical system $R$ is uncorrelated with the system $M$, we have the inequality
\begin{equation}
  \exp{\frac{S(C|M)}{n}} \geq \exp{\frac{S(A|M)}{n}} + \exp{\frac{S(\rho_R)}{n}}\ .
\end{equation}
\end{theorem}
\begin{remark}
  An important case for applications is the case when $R$ has the Gaussian probability density function $f_{Z,t} = \exp\left(-\frac{\|\xi\|^2}{2t}\right)/t^n$. In this special case the inequality reads
  \begin{equation}
    \exp{\frac{S(C|M)}{n}} \geq \exp{\frac{S(A|M)}{n}} + et\ .
  \end{equation}
\end{remark}
\begin{proof}
We define the evolution
\begin{equation}
  \rho_{ARM}(t) = \left(\cN(\lambda t) \otimes \cN_{\mathrm{cl}}((1-\lambda)t)\otimes \id_M \right)(\rho_{ARM})\ .\\
\end{equation}
Then, by compatibility with the heat semigroup, this amounts to an evolution of the $C$ system given by
\begin{equation}
\rho_{CM}(t) = \left(\cN(t) \otimes \id_M \right) (\rho_{CM})\ .
\end{equation}
This evolution preserves the condition $I(A:R|M) = 0$ because of the data-processing inequality. We also define
\begin{equation}
\phi(t) = S(C|M)_{\rho_{CM}(t)} - \lambda S(A|M)_{\rho_{AM}(t)} - (1-\lambda) S(R|M)_{\rho_{RM}(t)} \ .
\end{equation}
Then we have, because of the de Bruijn identity and compatibility with the heat semigroup as well as the Stam inequality,
\begin{equation}
\phi'(t) = J(C|M)_{\rho_{CM}(t)} - \lambda^2 J(A|M)_{\rho_{AM}(t)} - (1-\lambda)^2 J(R|M)_{\rho_{CM}(t)} \leq 0\ .
\end{equation}
Since $\phi$ is a linear combination of continuous concave functions, we have for $t \geq 0$,
\begin{equation}
\phi(t) - \phi(0) = \int_{0}^t \phi'(s) \text{d}s \leq 0.
\end{equation}
Using the universal scaling, we obtain
\begin{align}
\phi(0) &\geq \lim_{t \rightarrow \infty} \phi(t) \\
&= \lim_{t \rightarrow \infty} \left(S(C|M)_{\rho_{CM}(t)} - \lambda S(A|M)_{\rho_{AM}(t)} - (1-\lambda) S(R|M)_{\rho_{RM}(t)} \right)\\
&= n(-\lambda \log \lambda - (1-\lambda)\log(1-\lambda)) \ .
\end{align}
The theorem follows.
\end{proof}
\section{Optimality of the quantum conditional entropy power inequality}
This section is dedicated to the study of the optimality of the quantum conditional entropy power inequality stated in \autoref{thm:mainthm}. We show the following theorem:
\label{sec:tightness}
\begin{theorem}[optimality of the conditional entropy power inequality]
For any $a,b \in \mathbb{R}$ there exists a sequence of states $\left\{\rho^{(k)}_{AM}\right\}_{k \in \mathbb{N}}$ and a probability density function
$f: \mathbb{R}^2 \rightarrow \mathbb{R}$ such that the classical system $R$ is uncorrelated with $M$ and
\begin{align}
  \lim_{k\rightarrow\infty} S(A|M)_{\rho^{(k)}_{AM}} = a\ ,\qquad
S(R|M)_{f} = b\ ,
\end{align}
as well as
\begin{equation}
  \lim_{k\rightarrow\infty} \exp{S(C|M)_{\rho^{(k)}_{CM}}} = \exp a + \exp b\ ,
\end{equation}
where $\rho^{(k)}_{CM} = (\cE_f \otimes \id_{M})(\rho_{AM})$ with $\cE_f(\rho_A) = f \star \rho_A$.
\label{thm:tightness}
\end{theorem}
\begin{proof}
Let $\sigma^{(k)}_{AM}$ be the Gaussian state with the covariance matrix
\begin{equation}
  \Gamma^{(k)}_{AM} = \begin{pmatrix}k^2 & 0 & \sqrt{k^4-\frac{1}{4}} & 0\\ 0 & k^2 & 0 & -\sqrt{k^4-\frac{1}{4}}\\ \sqrt{k^4 - \frac{1}{4}} & 0 & k^2 & 0\\ 0 & -\sqrt{k^4-\frac{1}{4}} & 0 & k^2 \end{pmatrix}\ .
\end{equation}
Applying the heat semigroup on the quantum system $A$, we obtain the state $(\cN(t) \otimes \id_M)(\sigma^{(k)}_{AM})$ which has the covariance matrix
\begin{equation}
  \Gamma^{(k)}_{AM}(t) = \begin{pmatrix}k^2+t & 0 & \sqrt{k^4-\frac{1}{4}} & 0\\ 0 & k^2+t & 0 & -\sqrt{k^4-\frac{1}{4}}\\ \sqrt{k^4 - \frac{1}{4}} & 0 & k^2 & 0\\ 0 & -\sqrt{k^4-\frac{1}{4}} & 0 & k^2 \end{pmatrix}\ .
\end{equation}
The symplectic eigenvalues of this covariance matrix are
\begin{equation}
  \nu_{\pm}^{(k)}(t) = \frac{1}{2}\sqrt{4k^2 t \pm 2 t\sqrt{4k^2 t + t^2 + 1} + 2t^2 + 1} = k\sqrt{t} + \cO(1) \qquad (k \rightarrow\infty)\ .
\end{equation}
Hence we have
\begin{align}
 S(AM)_{(\cN(t)\otimes\id_M)(\sigma^{(k)}_{AM})} &= g\left(\nu_+-\frac{1}{2}\right) + g\left(\nu_- -\frac{1}{2}\right) \\
 &= \log k^2 + \log t + 2 + \cO\left(\frac{1}{k^2}\right)\ , \\
 S(M)_{(\cN(t)\otimes\id_M)(\sigma^{(k)}_{AM})} &= g\left(k^2-\frac{1}{2}\right) = \log k^2 + 1 + \cO\left(\frac{1}{k^4}\right)\ .
\end{align}
It follows that
\begin{equation}
 \lim_{k \rightarrow \infty} S(A|M)_{(\cN(t)\otimes\id_M)(\sigma^{(k)}_{AM})} = 1 + \log t\ .
\end{equation}
We now choose $\rho^{(k)}_{AM} = (\cN(e^{a-1}) \otimes \id_M)(\sigma^{(k)}_{AM})$, which fulfills
\begin{equation}
 \lim_{k \rightarrow \infty} S(A|M)_{(\rho^{(k)}_{AM})} = a\ .
\end{equation}
We further choose the classical system $R$ to be uncorrelated with $M$ and have probability density function
\begin{equation}
f = f_{Z,e^{b-1}} = \frac{e^{-\frac{\|\xi\|^2}{2e^{b-1}}}}{e^{b-1}}
\end{equation}
of a Gaussian with covariance matrix $e^{b-1}\id_2$. Then we have for the entropy
\begin{equation}
 S(R|M)_f = \log\left(e e^{b-1}\right) = b\ .
\end{equation}
The state $\rho^{(k)}_{CM}$ has the covariance matrix
\begin{equation}
  \Gamma^{(k)}_{CM} = \begin{pmatrix}k^2+e^{a-1}+e^{b-1} & 0 & \sqrt{k^4-\frac{1}{4}} & 0\\ 0 & k^2+e^{a-1}+e^{b-1} & 0 & -\sqrt{k^4-\frac{1}{4}}\\ \sqrt{k^4 - \frac{1}{4}} & 0 & k^2 & 0\\ 0 & -\sqrt{k^4-\frac{1}{4}} & 0 & k^2 \end{pmatrix}\ .
\end{equation}
Analogously to the calculation above, we obtain now
\begin{equation}
  \lim_{k\rightarrow\infty}S(C|M)_{\rho^{(k)}_{CM}} = 1+\log\left(e^{a-1}+e^{b-1}\right) = \log\left(e^a + e^b\right)\ ,
\end{equation}
and finally
\begin{equation}
  \lim_{k\rightarrow\infty}\exp{S(C|M)_{\rho^{(k)}_{CM}}} = e^a + e^b\ .
\end{equation}
\end{proof}
\section{Applications}
\label{sec:apps}
The quantum conditional entropy power inequality~\eqref{eq:cqcondepi} has various applications in the derivation of information-theoretic inequalities. We are going to show a variety of results regarding quantum conditional entropies. Many of these results have direct analogs in the case of unconditioned quantum entropies as well as in classical information theory.
\subsection{Isoperimetric inequality for conditional entropies}
\begin{lemma}[Quantum conditional Fisher information isoperimetric inequality]
\label{lem:fisherisoperi}
\begin{equation}
  \frac{\mathrm{d}}{\mathrm{d}t} \left[ \frac{1}{n} J(A|M)_{(\cN(t)\otimes \id_M)(\rho_{AM})} \right]^{-1} \bigg|_{t=0} \geq 1\ .
  \label{eq:fisheriso}
\end{equation}
\end{lemma}
\begin{proof}
  We note that $(\cN(t)\otimes \id_M)(\rho_{AM}) = ( \cE_{f_{Z,t}}\otimes \id_M)(\rho_{AM})$, where again $\cE_{f}(\rho_A) = f\star \rho_A$ and $f_{Z,t}(\xi) = \exp\left(-\frac{\|\xi\|^2}{2t}\right)/t^n$. Applying the conditional Stam inequality \eqref{eq:stam}, we obtain
\begin{equation}
  \left(J(A|M)_{(\cN(t)\otimes \id_M)(\rho_{AM})}^{-1} - J(A|M)_{\rho_A}^{-1} \right) \geq J(R|M)_{f_{Z,t}}^{-1} = \frac{t}{n} \ .
\end{equation}
This implies
\begin{equation}
  \frac{1}{t} \left(J(A|M)_{(\cN(t)\otimes \id_M)(\rho_{AM})}^{-1} - J(A|M)_{\rho_{AM}}^{-1} \right) \geq \frac{1}{n}\ .
\end{equation}
Taking the limit $t\rightarrow 0$ implies the result.
\end{proof}
\begin{theorem}[Isoperimetric inequality for quantum conditional entropies]
\begin{equation}
\frac{1}{n}J(A|M)_{\rho_{AM}} \exp{\frac{S(A|M)_{\rho_{AM}}}{n}} \geq e\ .
\end{equation}
\end{theorem}
\begin{proof}
We apply the conditional de Bruijn identity \cite[Eq. 63]{depalmaconditional} and see that
\begin{equation}
  \frac{\mathrm{d}}{\mathrm{d}t} \exp{\frac{S(A|M)_{(\cN(t)\otimes \id_M)(\rho_{AM})}}{n} }\bigg|_{t=0} = \frac{1}{n}J(A|M)_{\rho_{AM}} \exp{\frac{S(A|M)_{\rho_{AM}}}{n}}\ .
\end{equation}
Recalling once again that $(\cN(t)\otimes \id_M)(\rho_{AM}) = (\cE_{f_{Z,t}}\otimes \id_M)(\rho_{AM})$ and inserting this into the conditional entropy power inequality \eqref{eq:cqcondepi} yields
\begin{equation}
  \exp{\frac{S(A|M)_{(\cN(t)\otimes\id_M)(\rho_{AM})}}{n}} - \exp{\frac{S(A|M)}{n}} \geq \exp{\frac{S(R|M)}{n}} = et\ .
\end{equation}
Dividing this equation by $t$ and taking the limit $t \rightarrow 0$ concludes the proof of the theorem.
\end{proof}
\subsection{Concavity of the quantum conditional entropy power along the heat flow}
\begin{theorem}[Concavity of the quantum conditional entropy power along the heat flow]
  \label{thm:concavityheat}
\begin{equation}
  \frac{\mathrm{d}^2}{\mathrm{d}t^2} \exp{\frac{S(A|M)_{(\cN(t)\otimes \id_M)(\rho_{AM})}}{n} }\bigg|_{t=0} \leq 0\ .
\end{equation}
\end{theorem}
\begin{proof}
We write here $P(t) = \exp{\frac{S(A|M)_{(\cN(t)\otimes\id_M)(\rho_{AM})}}{n}}$ and apply the de Bruijn identity \cite[Eq. 63]{depalmaconditional} twice to obtain
\begin{equation}
  \frac{\mathrm{d}^2}{\mathrm{d}t^2}P(t) \bigg|_{t=0} = P(0)\left(\left[\frac{1}{n} J(A|M) \right]^2 + \frac{1}{n} \frac{\mathrm{d}}{\mathrm{d}t} J(A|M)_{(\cN(t)\otimes\id_M)(\rho_{AM})} \bigg|_{t=0}\right)\ .
\end{equation}
The quantum conditional Fisher information isoperimetric inequality stated in \eqref{eq:fisheriso} can be restated as
\begin{equation}
  \frac{1}{n^2} J(A|M)^2 + \frac{1}{n}\frac{\mathrm{d}}{\mathrm{d}t} J(A|M)_{(\cN(t)\otimes \id_M)(\rho_{AM})} \big|_{t=0}\leq 0\ .
\end{equation}
Since $P(0) \geq 0$, the concavity of the quantum conditional entropy power follows.
\end{proof}
\subsection{Converse bound on the entanglement-assisted classical capacity for a non-Gaussian classical noise channel}
\label{sec:ea}
The fact that conditional entropy power inequalities imply upper bounds on the entanglement-assisted classical capacity has been known since the first quantum conditional entropy power inequality has been proposed~\cite{koenigconditionalepi2015,depalmaconditional}.
In this section, we use the conditional entropy power inequality \eqref{eq:cqcondepi} to prove such an upper bound for a classical noise channel which is not necessarily Gaussian.

We consider the classical noise channel with a given (possibly non-Gaussian) noise probability density function $f: \mathbb{R}^{2n} \rightarrow \mathbb{R}$. This channel is given by $\cE_f: A \rightarrow C$,
\begin{equation}
\cE_{f}(\rho_A) = f \star \rho_A\ .
\end{equation}
The entanglement-assisted classical capacity~\cite{holevobook,Wilde_2017,wildeqi} is then
\begin{equation}
C_{\mathrm{ea}}(\cE_f) = \sup\left\{ I(C:M)_{(\cE_f\otimes\id_M)(\rho_{AM})}: \rho_{AM} \text{ pure}, \tr_A[H_A \rho_A]\leq nE \right\}\ .
\end{equation}
The energy constraint $\tr_A[H_A\rho_A] \leq nE$ amounts to the assumption that the sender can only use states of a finite average energy $E$ per mode. This assumption is required to make the entanglement-assisted capacity finite. Indeed, the assumption that a sender can use an unlimited amount of energy is unphysical.
Let
\begin{equation}
  E_0 := \frac{E(f)}{2n}, \qquad S_0 := \frac{S(f)}{n}
\end{equation}
be the average energy and entropy per mode of $f$. We can now bound the maximum output entropy as
\begin{align}
  \sup_{\substack{\tr[H_A\rho_A] \leq nE }} S\left(\cE_f(\rho_A)\right) &= \sup_{\substack{\tr[H_A\rho_A]\leq nE}} S\left(\cE_f(\rho_{A,0}^{(d(\rho_A))}) \right)\\
  &= \sup_{\substack{\tr[H_A\rho_A] \leq nE}} S\left(\cE_f(\rho_{A,0})^{(d(\rho_A))} \right) \\
  &= \sup_{\substack{\tr[H_A\rho_A] \leq nE}} S\left(\cE_f(\rho_{A,0})\right)\\
  &= \sup_{\substack{\tr[H_A\rho_{A}] \leq nE\\ d(\rho_{A}) = 0}} S\left(\cE_f(\rho_{A})\right)\ ,
\end{align}
where we have written $\rho_{A,0} = D(-d(\rho_A))\rho_A D(-d(\rho_A))^\dagger$ for the state $\rho_A$ which has been displaced by its first moments such that it is centered, i.e., the first moments of $\rho_{A,0}$ are zero. The first equality follows by this definition. In the second equality we have used compatibility of the convolution \eqref{eq:cqconv} with displacements, and in the third equality, we have used the fact that the von Neumann entropy is invariant under conjugation with unitaries. In the fourth equality, we have used the fact that
\begin{align}
  \tr[H_A \rho_{A,0}] &= \tr\left[\left(\frac{1}{2}\sum_{k=1}^{2n}R_k^2 - \frac{n}{2}\id_A \right)\rho_{A,0} \right]\\
  & = \tr\left[\left(\frac{1}{2} \sum_{k=1}^{2n} D(-d(\rho_A))^\dagger R_k^2 D(-d(\rho_A)) - \frac{n}{2}\id_A \right) \rho_A\right]\\
  & = \tr\left[\left(\frac{1}{2} \sum_{k=1}^{2n} \left(R_k - d_k(\rho_A)\right)^2 - \frac{n}{2}\id_A\right)\rho_A \right]\\
  & = \tr\left[\left(\frac{1}{2} \sum_{k=1}^{2n} R_k^2 - \frac{n}{2}\id_A\right)\rho_A\right] - \sum_{k=1}^{2n} d_k(\rho_A) \tr[R_k \rho_A]+ \frac{1}{2}\sum_{k=1}^{2n} d_k(\rho_A)^2\\
  & = \tr[H_A \rho_{A}] - \frac{1}{2} \left\| d(\rho_{A})\right\|^2 \leq \tr[H_A\rho_A] \leq nE\ .
\end{align}
Therefore in order to upper bound the output entropy, we can restrict our consideration to centered states, i.e., states which have zero first moments. The average energy per mode at the output $\cE_f(\rho_A)$ is then bounded as
\begin{align}
  \frac{1}{n}\tr_C[H_C \cE_f(\rho_A)] &= \tr\left[ \left(\frac{1}{2n}\sum_{k=1}^{2n} R_k^2 - \frac{\id}{2}\right)(f\star \rho_A) \right]\\
  &= \tr\left[\int_{\mathbb{R}^{2n}}f(\xi) \left(\frac{1}{2n}\sum_{k=1}^{2n} R_k^2 - \frac{\id}{2} \right) D(\xi) \rho_A D(\xi)^\dagger \frac{\mathrm{d}^{2n}\xi}{(2\pi)^n} \right]\\
  &= \tr\left[\int_{\mathbb{R}^{2n}}f(\xi) \left(\frac{1}{2n}\sum_{k=1}^{2n} (R_k + \xi_k)^2 - \frac{\id}{2}\right) \rho_A\frac{\mathrm{d}^{2n}\xi}{(2\pi)^n} \right]\\
  &= \tr\left[\left(\frac{1}{2n}\sum_{k=1}^{2n} R_k^2 - \frac{\id}{2} \right)\rho_A \right] + \frac{1}{2n} \sum_{k=1}^{2n} \int_{\mathbb{R}^{2n}} f(\xi) \xi_k^2 \frac{\mathrm{d}^{2n}\xi}{(2\pi)^n}\\
  &\qquad + \frac{1}{n}\sum_{k=1}^{2n}\int_{\mathbb{R}^{2n}} f(\xi) \xi_k \tr[R_k\rho_A] \frac{\mathrm{d}^{2n}\xi}{(2\pi)^n}\\
  &= \frac{1}{n} \tr[H_A \rho_A] + E_0 \leq E + E_0\ .
\end{align}
In the last line we have used that $\tr[R_k\rho_A] = 0$ by assumption.
Hence by the fact that thermal states maximize the von Neumann entropy among all states with a given average energy, we have that the maximum output entropy
is bounded by
\begin{equation}
S(\cE_f(\rho_A)) \leq n g\left(E + E_0\right)\ .
\end{equation}
From the conditional entropy power inequality \eqref{eq:cqcondepi}, we obtain
\begin{align}
  \exp{\frac{S(C|M)}{n}} &\geq \exp{\frac{S(A|M)}{n}} + \exp{S_0}\\
  &= \exp{\frac{-S(A)}{n}} + \exp{S_0} \\
  &\geq \exp{\left(-g(E)\right)} + \exp{S_0}\ .
\end{align}
This implies for the mutual information
\begin{align}
I(C:M) &= S(\cE_f(\rho_A)) - S(C|M)_{(\cE_f\otimes\id_M)(\rho_{AM})}\\
&\leq n g\left(E + E_0 \right) - n \log\left( e^{-g(E)} + e^{S_0} \right) \ .
\end{align}
Therefore, for the entanglement-assisted classical capacity, we have the upper bound
\begin{equation}
C_{\mathrm{ea}}(\cE_f) \leq n g\left(E + E_0 \right) - n \log\left( e^{-g(E)} + e^{S_0} \right) \ .
\end{equation}
\section{A simple proof of convergence rate of the quantum Ornstein-Uhlenbeck semigroup}
\label{sec:qOU}
We consider the quantum Ornstein-Uhlenbeck semigroup which is a one-parameter semigroup of CPTP maps $\{\cP^{(\mu,\lambda)}(t) = e^{t\cL_{\mu,\nu}} \}_{t\geq 0}$ on the one-mode Gaussian quantum system $A$ generated by the Liouvillian
\begin{equation}
\cL_{\mu,\lambda} = \mu^2 \cL_- + \lambda^2 \cL_+ \qquad \text{ for } \mu > \lambda > 0\ ,
\end{equation}
where
\begin{equation}
\cL_+(\rho) = a^\dagger \rho a - \frac{1}{2} \{ aa^\dagger, \rho \} \qquad \text{ and } \qquad \cL_-(\rho) = a\rho a^\dagger - \frac{1}{2} \{ a^\dagger a, \rho\}\ ,
\end{equation}
and $a$ is the ladder operator of $A$.

The map $\cP^{(\mu,\lambda)}(t)$ is equivalent to a beam splitter with transmissivity $\eta = e^{-(\mu^2-\lambda^2)t}$ and environment state $\omega^{(\mu,\lambda)} := \frac{\mu^2-\lambda^2}{\mu^2} \sum_{k=0}^\infty \left(\frac{\lambda^2}{\mu^2}\right)^k\proj{k}$. This is a Gaussian thermal state with the covariance matrix equal to $\Gamma^{(\mu,\lambda)} = \frac{1}{2} \frac{\lambda^2+\mu^2}{\mu^2-\lambda^2}\id_2.$
The state $\omega^{(\mu,\lambda)}$ is also the unique fixed point of the quantum Ornstein-Uhlenbeck semigroup with parameters $\mu$ and $\lambda$. It is known that the qOU semigroup converges in relative entropy to the fixed point at an exponential rate given by the exponent  $\mu^2-\lambda^2$:
\begin{equation}
D\left(\cP^{(\mu,\lambda)}(t)(\rho) \big\| \omega^{(\mu,\lambda)} \right) \leq e^{-(\mu^2-\lambda^2)t}D\left(\rho \big\| \omega^{(\mu,\lambda)}\right) \qquad \text{ for all } t \geq 0\ .
\end{equation}
This is a conjecture stated in \cite{HubKoeVersh16}, which was proven in \cite{Carlen_2017} using methods of gradient flow.

Here we want to study a slightly different, more general scenario: We consider a bipartite quantum system $AM$, where the system $A$ undergoes a qOU evolution.
We are going to show a similar convergence statement in this situation, namely, that the system converges in relative entropy to the product state $\omega_A^{(\mu,\lambda)}~\otimes~\tr_A(\rho_{AM})$ at an exponential rate.

\begin{theorem}
  \label{thm:qOUbipartite}
We have for any quantum state $\rho_{AM}$
\begin{equation}
  D\left((\cP^{(\mu,\lambda)}(t)\otimes \id_M)(\rho_{AM}) \big\| \omega_A^{(\mu,\lambda)}\otimes\rho_M \right) \leq e^{-(\mu^2-\lambda^2)t} D\left( \rho_{AM} \big\| \omega_A^{(\mu,\lambda)} \otimes \rho_M \right)\ ,
\end{equation}
where $\rho_M = \tr_A (\rho_{AM})$ is the marginal state of $\rho_{AM}$ on the system $M$. In particular, Eq. \eqref{eq:qOUrate} holds.
\end{theorem}
\begin{proof}
Write $\rho_{AM}(t) = (\cP_{\mu,\lambda}(t) \otimes\id_M)(\rho_{AM})$, we then have
\begin{align}
D\left( \rho_{AM}(t) \big\| \omega_A^{(\mu,\lambda)} \otimes \rho_M\right) &= -S(A|M)_{\rho_{AM}(t)} - \tr\left( \rho_A(t) \log \omega_A^{(\mu,\lambda)} \right)\\
&\leq  -\eta S(A|M)_{\rho_{AM}} - \left(1-\eta\right)S(\omega^{(\mu,\lambda)}_A)\\
&  -\eta\tr\left(\rho_A\log\omega_A^{(\mu,\lambda)} \right) - (1-\eta)\tr\left(\omega_A^{(\mu,\lambda)}\log\omega_A^{(\mu,\lambda)} \right)\\
&= e^{-(\mu^2-\lambda^2)t}D\left(\rho_{AM}\big\| \omega_A^{(\mu,\lambda)}\otimes \rho_M\right)\ .
\end{align}
\end{proof}
This implies exponential convergence to the fixed point both on bipartite systems as well as the result \eqref{eq:qOUrate}.

\section{Conclusion}
\label{sec:conclusion}
We have established a conditional entropy power inequality for classical noise channels in bosonic quantum systems, modeled by the convolution \eqref{eq:cqconv}. This inequality implies the unconditioned entropy power inequality for this convolution and lifts regularity problems in previous proofs in this area.
In the conditioned case, this inequality is optimal, while the optimal inequality in the unconditioned case remains unsolved.
This situation is analogous to the situation for the beam splitter \cite{HubKoe17}, where the optimal unconditioned inequality is conjectured to be the entropy photon-number inequality~\cite{Guhaetal07}, which states that couples of thermal Gaussian input states minimize the output entropy of the beam splitter among all the couples of independent input states, each with a given entropy.
The entropy photon-number inequality has been recently proven for the one-mode beam splitter in the particular case where one of the two inputs is a thermal Gaussian state \cite{De_Palma_2016passive,De_Palma_2017attenuator,de2017gaussian,qi2016thermal,De_Palma_2016_maximizers,De_Palma_2017} and in some very special cases for the multi-mode beam splitter \cite{de2017multimode,GiovannettietalGaussianoutput}, and it otherwise remains an open challenging conjecture (see~\cite{doi:10.1063/1.5038665} for a review).
Similarly, an analogous optimal inequality has been conjectured for the quantum additive noise channel~\cite{HubKoe17}. While the validity of this inequality remains an open problem (besides the special case covered in~\cite{De_Palma_2017}), the conditional entropy power inequality proven in this paper is optimal and settles the problem in the presence of quantum memory.

We have used our new conditional entropy power inequality to provide upper bounds on the entanglement-assisted classical capacity of quantum non-Gaussian additive noise channels, and to prove conditional quantum versions of various celebrated results from geometric analysis. Moreover, we have shown how conditional entropy power inequalities can be used to study the convergence rate
of quantum dynamical semigroups, giving a simple and short proof of the exponential convergence of the quantum Ornstein-Uhlenbeck semigroup in relative entropy.

\section*{Acknowledgments}
GdP thanks Dario Trevisan for useful comments.

GdP acknowledges financial support from the European Research Council (ERC Grant Agreements Nos. 337603 and 321029), the Danish Council for Independent Research (Sapere Aude), VILLUM FONDEN via the QMATH Centre of Excellence (Grant No. 10059), and the Marie Sk\l odowska-Curie Action GENIUS (Grant No. 792557).
SH is supported by the Technische Universität München - Institute for Advanced Study, funded by the German Excellence Initiative and the European Union Seventh Framework Programme under grant agreement no. 291763. SH acknowledges additional support by DFG project no. K05430/1-1.

\includegraphics[width=0.05\textwidth]{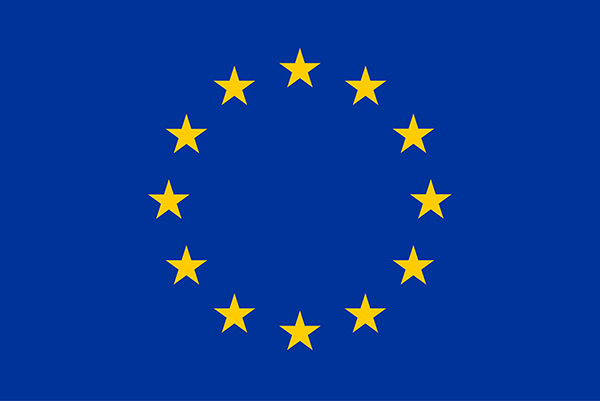}
This project has received funding from the European Union's Horizon 2020 research and innovation programme under the Marie Sk\l odowska-Curie grant agreement No. 792557.


\end{document}